\documentclass[letterpaper,10pt,conference]{ieeeconf}  

\IEEEoverridecommandlockouts                              

\usepackage{graphics} 
\usepackage{epsfig} 
\usepackage{amsmath} 
\usepackage{amssymb}  
 
\usepackage{amsthm}
\usepackage{cite}
\usepackage{color}
\usepackage{enumerate}
\usepackage{graphicx}
\usepackage{mathtools}
\usepackage{cancel}
\usepackage{url}
\usepackage{algorithm}
\usepackage{algpseudocode}
\usepackage{textgreek}

\graphicspath{{figures/}}

\newcommand{\naturals}{\mathbb{N}}
\newcommand{\real}{\mathbb{R}}
\newcommand{\realnonneg}{\mathbb{R}_{\ge 0}}
\newcommand{\realpos}{\mathbb{R}_{> 0}}

\newcommand{\map}[3]{#1:#2 \rightarrow #3}

\newcommand{\longthmtitle}[1]{\mbox{}{\textit{(#1):}}}

\newcommand{\setdefb}[2]{\big\{#1 \; | \; #2\big\}}

\newcommand*{\SetSuchThat}[1][]{} 
\newcommand*{\MvertSets}{%
    \renewcommand*\SetSuchThat[1][]{%
        \mathclose{}%
        \nonscript\;##1\vert\penalty\relpenalty\nonscript\;%
        \mathopen{}%
    }%
}
\MvertSets 
\DeclarePairedDelimiterX \Set [2] {\lbrace}{\rbrace}
    {\,#1\SetSuchThat[\delimsize]#2\,}

\newcommand{\Ac}{\mathcal{A}}
\newcommand{\Cc}{\mathcal{C}}
\newcommand{\Dc}{\mathcal{D}}

\newcommand{\Sc}{\mathcal{S}}
\newcommand{\Xc}{\mathcal{X}}

\newcommand{\mydelta}{\textbf{\textdelta}}

\newtheorem{proposition}{Proposition}

\theoremstyle{definition}
\newtheorem{definition}{Definition}

\newtheorem{remark}{Remark}

\newtheorem{assumption}{Assumption}
\newtheorem{problem}{Problem}

\newcommand{\sgn}{{\operatorname{sgn}}}
\newcommand{\tgt}{{\operatorname{tgt}}}

\newcommand{\Lie}{\mathcal{L}}

\newcommand{\Exp}{\mathbb{E}}
\renewcommand{\baselinestretch}{1}
\allowdisplaybreaks

\DeclareMathOperator*{\argmax}{argmax}

\title{Hierarchical Event-Triggered Systems:\\ Safe Learning of Quasi-Optimal Deadline Policies}
\author{Pio Ong, Manuel Mazo Jr., and  Aaron D. Ames
\thanks{This research is supported in part by TII under project \#A6847.}
  \thanks{Pio Ong and Aaron D. Ames are with the Department of Mechanical and Civil Engineering, California Institute of Technology, Pasadena, CA 91125, USA. {\tt\small \{pioong,ames\}@caltech.edu}}%
  \thanks{Manuel Mazo Jr. is with the Delft Center for Systems and Control, Faculty of Mechanical Engineering, Delft University of Technology, The Netherlands. {\tt\small m.mazo@tudelft.nl}}
}

\begin{document}

\maketitle

\begin{abstract}
    We present a hierarchical architecture to improve the efficiency of event-triggered control (ETC) in reducing resource consumption. This paper considers event-triggered systems generally as an impulsive control system in which the objective is to minimize the number of impulses. Our architecture recognizes that traditional ETC is a greedy strategy towards optimizing average inter-event times and introduces the idea of a deadline policy for the optimization of long-term discounted inter-event times. A lower layer is designed employing event-triggered control to guarantee the satisfaction of control objectives, while a higher layer implements a deadline policy designed with reinforcement learning to improve the discounted inter-event time. We apply this scheme to the control of an orbiting spacecraft, showing superior performance in terms of actuation frequency reduction with respect to a standard (one-layer) ETC while maintaining safety guarantees.
\end{abstract}

\section{Introduction}

A common goal in the development of networked control systems is the reduction of communications bandwidth needed to close the control loop. Similarly, in other applications it is desirable to reduce the amount of actuation changes---to prevent actuators wear or reduce limited actuation resources use. An example of the latter can be found in space applications where fuel availability is limited. The study of how to reduce feedback updates, communication and/or actuation, in control systems has received a lot of attention in the past couple of decades. Much of the work has focused on the minimum data-rates necessary to stabilize and guarantee performance of control loops, see, e.g.~\cite{DL-JPH:05, ST-SM:04, GNN-FF-SZ-RJE:07}. Alternatively, methods that deviate from the periodic control paradigm have been considered~\cite{KJA-BMB:02, MV-MF-PM:03}. In particular, Lyapunov based event-triggered control (ETC) and self-triggered control (STC), pioneered in~\cite{PT:07, AA-PT:10}, emerged as paradigms attracting much attention from the control community. 

The basic idea behind ETC is to monitor (a proxy for) the evolution of certificate functions, e.g., Lyapunov or barrier functions certifying stability/safety, and performance of the control system to determine when feedback updates are required. In STC, rather than continuously monitoring the certificates, the approach is to predict its evolution based on the last received measurements and determine beforehand when controller updates will be required. For a more detailed introduction to these topics see~\cite{WPMHH-KHJ-PT:12}. While the initial focus of this line of research was on stabilization, recent extensions have extended the ETC principles to guarantee safety~\cite{PO-JC:24-tac, AJT-PO-JC-AA:21-csl, PO-ADA:23-cdc} employing barrier functions. Two major challenges faced by aperiodic control paradigms, as ETC and STC, are their scheduling and the analysis of feedback usage. Both of these issues are linked to predicting the possible inter-sampling patterns exhibited by such systems~\cite{ASK-MM:16,GG-MM:23, GD-MM:22}.

In~\cite{GG-MM:21} the observation was made that ETC, as well as most STC strategies, are greedy optimizers towards the goal of minimizing average feedback resources' usage. 
Analytically predicting and optimizing aperiodic sampling patters for control is an extremely challenging problem.
The authors of~\cite{GG-MM:21} circumvent the problem by proposing a STC strategy for LTI systems, based on symbolic abstractions and formal synthesis, resulting in aperiodic implementations that optimize the average inter-event times. Similarly, following a more heuristic approach~\cite{PO-ADA:23-cdc} aimed also at optimizing inter-sample times in the context of safe satellite control.

Hierarchical, also called layered, approaches are a common approach in control engineering to address complex problems \cite{NM-ADA-JCD:24}. The maxim applied is that of divide-and-conquer, splitting a complex task into smaller/easier to address tasks at different levels of abstraction. This approach appears in biological systems~\cite{JM-MB-GW:19}, and has been extensively applied in robotics~\cite{LS-OK:05, JK-RTF-VRK-ADA-KAH:23}, often also taking advantage of the layering to produce multi-rate implementations~\cite{UR-AS-ADA:22}.
A remaining problem in ETC is how to optimize long-term metrics of resource usage, usually as a function of inter-event times, while retaining performance and safety guarantees, particularly when the dynamics are not linear.

In this paper, we draw inspiration from hierarchical structures, and the preliminary work in~\cite{RR:22}, to propose a layered approach to ETC for general nonlinear systems that optimizes long-term metrics of resources' use while providing strong safety/performance guarantees. The proposed architecture relies on a lower layer providing hard guarantees via well established ETC designs, combined with a supervisory control layer on top that optimize the discounted inter-event times. In particular, we build on Q-learning~\cite{PD-CJCHW:92} and propose a faster algorithm (by exploiting the structure of our problem) in the top layer to learn how to prescribe deadlines for the lower layer that enforce feedback updates possibly earlier than the lower layer ETC mechanism prescribes. The proposed architecture thereby enables \textit{safe learning} of strategies, including online optimization during safe operation.  We demonstrate the effectiveness of our approach on the application domain of safe control of orbiting spacecraft.

\section{Event-triggered Control Preliminaries}
We begin by reviewing key concepts in event-triggered control and establish the motivation for the problem this paper addresses.

Event-triggered control (ETC) is a controller implementation strategy towards resource conservation. To illustrate the main ideas behind ETC, we review its use in the setting of sample-and-hold controller implementations. Consider a sample-and-hold control system\footnote{Throughout this paper, we use the notation $\real$, $\realnonneg$, $\realpos$ for the set of real, nonnegative real, and positive real, respectively. For a vector $x\in\real^n$, $\|x\|$ denotes its Euclidean norm. We use $\Lie_fh = \left.\frac{\partial h}{\partial x}\right|_xf(x)$ for the Lie derivative of the function $\map{h}{\real^n}{\real}$ along the vector field $\map{f}{\real^n}{\real^n}$.}:
\begin{subequations}\label{sys:sample-and-hold}
    \begin{align}
    \begin{bmatrix}
        \dot \eta \\ \dot u
    \end{bmatrix} &= \begin{bmatrix}F(\eta,u)+d\\ 0 \end{bmatrix}\\
     \label{sys:sample-and-hold_jump}
    \begin{bmatrix}
        \eta^+ \\ u^+
    \end{bmatrix} &= \begin{bmatrix}
        \eta \\ u
    \end{bmatrix} + \underbrace{\begin{bmatrix}
        0 \\ \kappa(\eta)-u
    \end{bmatrix}}_{=\Delta u},~ t\in\{t_i\}_{i\in\naturals}^\infty
    \end{align}
\end{subequations}
with system state $\eta(t)\in\real^q$ and control input variable $u(t)\in\real^m$. The state evolves along the vector field $\map{F}{\real^q \times \real^m}{\real^q}$, subject to a disturbance $d(t)\in\real^q$. The motivation for sample-and-hold implementations is that the state-feedback $u(t)=\kappa(x(t))$ with a controller $\map{\kappa}{\real^q}{\real^m}$ cannot be updated in a continuous manner on a digital platform. Rather, common practice involves sampling the values of the controller periodically at different time instants $\{t_i\}^\infty_{i\in\naturals}$ at a high enough frequency to ensure that the sample-and-hold system behaves similarly to the one with continuous feedback. In such implementations, the jump map \eqref{sys:sample-and-hold_jump} occurs at regular intervals.

In many settings, it is undesirable to make frequent changes to the control, e.g., due to communication constraints, or actuation limitations. To this end, ETC arises as an alternative approach for controller sampling. ETC aims at reducing how often the controller is sampled and updated -- how often the jump map~\eqref{sys:sample-and-hold_jump} is executed -- while ensuring desirable behaviors, e.g. stability and safety. The approach consists in monitoring the system behavior and only updating the controller when necessary. In the case of stabilization problems without any disturbance ($d\equiv 0)$, for example, ETC sampling schemes aim at guaranteeing a Lyapunov function $\map{W}{\real^n}{\real}$ decreases along trajectories by determining the update times according to:
\begin{equation}
    t_{i+1} = \min\setdefb{t\geq t_i}{\Lie_FW(\eta(t),u(t))\geq 0},
\end{equation}
which enforces a controller update whenever the condition  $\Lie_FW(\eta(t),u(t))<0$ with $u(t) = \kappa(\eta(t_i))$ no longer holds. The ETC sampling scheme builds on holding each computed control input for as long as possible. In this sense, the scheme is a greedy approach towards the goal of avoiding overspending resources in control updates. Many works in the literature have demonstrated in practice that this heuristic approach reduces the resource usage in comparison to periodic sampling. Nevertheless, the shift to an aperiodic sampling raises the possibility of Zeno behavior -- infinite sampling instances within a finite period of time. As such, the drawback of ETC is that it requires additional analysis to establish a minimum inter-event time (MIET) in order to ensure it can be implemented in practice. 

\section{Problem description}
In this paper, we model event-triggered systems more generally as an impulsive control system with the flow:
\begin{subequations}
    \label{sys:ET}
    \begin{equation}
        \dot x = f(x)+d,
    \end{equation}
where $x(t)\in\mathcal{X}\subseteq\real^n$ is the system state. Here we take an impulsive modeling perspective of the event-triggered system. In this setting, the control input does not appear explicitly in the flow\footnote{For sample-and-hold systems~\eqref{sys:sample-and-hold}, the continuous control input~$u$ is included in the state $x=(\eta,u)$. Note $u$ does not depend on time $t$ but only on time instances $\{t_i\}_{i=0}^\infty$ in this formulation because $\dot u=0$}. Instead, we consider the effect of the controller to be completely captured in the jump dynamics: 
    \begin{align}
        &x(t_i^+) = g(x(t_i),v(t_i)),~ t\in\{t_i\}_{i\in\naturals}^\infty \label{sys:ET-jump}
    \end{align}
\end{subequations}
where $v(t)\in\mathcal{V}\subseteq\real^p$ represents an impulsive input variable that can instantly influence the state $x$ through the jump map~$\map{g}{\mathcal{X}\times\mathcal{V}}{\mathcal{X}}$. We are especially interested in scenarios where the actuation of $v$ implies the use of some scarce resource, which we would like to minimize. As a result, we wish the time instances $\{t_i\}_{i\in\naturals}^\infty$ to be as \textit{sparse} as possible. To this end, we assume the time sequence is determined iteratively using a triggering condition:
\begin{equation}\label{eq:trigger}
    t_{i+1} = \min\{\min\setdefb{t\geq t_i}{\Xi(x(t),x(t_i))\geq 0}, t_i+\delta_{\max} \},
\end{equation}
with an \textit{objective-based triggering condition}~$\map{\Xi}{\mathcal{X}\times\mathcal{X}}{\real}$. 
In~\eqref{eq:trigger}, the trigger deadline $\delta_{\max}$ denotes a maximum allowable time between events, sometimes called the system \textit{heartbeat}, used to enforce that a minimum amount of feedback is always present.
\begin{definition}\longthmtitle{Objective-based Triggering Condition}
    A triggering condition $\Xi$ is objective-based if its value along the trajectory being negative, $\Xi(x(t),x(t_i))<0$, at all time implies that the system control objectives, e.g., stability and safety, are met.\hfill~$\bullet$
\end{definition}
Note that the triggering condition could be generalized to be time-dependent and potentially depend on all the past states~$x$ and past impulsive inputs $v$. Here we simplify $\Xi$ to the form most commonly found in the literature. In addition, the formulation of the system thus far is a generalization of the sample-and-hold problem discussed earlier where the objective-based triggering condition is based on a Lyapunov condition for stability. 

Event-triggered control generally yields an aperiodic time sequence $\{t_i\}_{i\in\naturals}^\infty$. This makes it difficult to assess the sparsity of the time instances along trajectories and their variability across different initial conditions. In order to simplify the problem and make the presentation cleaner, we take the following assumption. 
\begin{assumption}\longthmtitle{Deterministic System}\label{assump:deterministic}
A state-feedback control policy $\map{\Pi}{\Xc}{\mathcal{V}}$ determining the control impulses $v=\Pi(x)$ is given. Furthermore, the disturbance $x\mapsto d(x)$ stems from model uncertainty and is state-dependent.~\hfill$\bullet$   
\end{assumption}

This assumption can be relaxed, and stochastic noise can be considered. The assumption allows us to consider inter-event times as being fully determined by the state $x(t_i)$ at the last actuation instance, i.e., there exists a (possibly not analytic) function $\map{\tau_\Xi}{\Xc}{\realnonneg}$ such that
$
    \tau_\Xi(x(t_i)) =  t_{i+1}-t_i.
$
This defines the minimum inter-event time (MIET):
\begin{equation}
\tau^*_\Xi = \inf_{x\in\real^n}\tau(x),
\end{equation}
which is the smallest time between two consecutive jump instances, from any possible system trajectory using triggering condition $\Xi$. Most works on ETC provide a MIET in order to rule out Zeno behavior.

\begin{assumption}\longthmtitle{Positive Minimum Inter-Event Time}\label{assump:MIET}
    Consider system~\eqref{sys:ET} using an objective-based triggering condition~$\Xi$ to determine iteratively the time instances $\{t_i\}_{i\in\naturals}^\infty$ as in \eqref{eq:trigger}. We assume the MIET $\tau^*_\Xi>0$ is positive.~\hfill$\bullet$ 
\end{assumption}

While the existence of a strictly positive MIET is useful to ensure the practicality of implementations, it is an insufficient metric to measure the efficiency of the implementation with respect to use of resources, in the sense that the MIET is a very myopic, worst-case, metric. A better metric for measuring resource consumption should measure long-term distributions of inter-event times. An example of such a metric is the average inter-event time (AIET), which, for a given initial condition $x_0$,  can be computed for the corresponding trajectory as:
\begin{equation}
    \bar \tau(x_0) 
    =\liminf_{s\rightarrow \infty} \frac{1}{s+1}\sum_{i=0}^s \tau_\Xi(x(t_i)).
\end{equation}
Although AIET is a good metric for resource consumption over time, it is more practical to place higher importance to earlier time due to the prediction inaccuracies that may grow over time from system model imperfection, always present in real applications. 
An alternative metric, that allows us to control the  weight given to transient sampling patterns, is the discounted inter-event time (DIET):
\begin{equation}
    \bar \tau_\gamma(x_0) =\liminf_{s\rightarrow \infty} \sum_{i=0}^s \gamma^i\tau_\Xi(x(t_i)).
\end{equation}
where $\gamma \in (0,1)$ is a discount factor. In this paper we consider DIET as the metric to optimize for the additional reason that it facilitates the implementation of learning approaches, in particular Q-learning, for its optimization. Nonetheless, this can be extended employing available alternatives~\cite{SM:96} for the optimization of average costs as the one defined by the AIET. The goal of our paper is to find strategies that maximize the DIET.
\begin{problem}\longthmtitle{Discounted Inter-Event Time Problem}\label{problem:diet}
    Consider the impulsive system~\eqref{sys:ET}. 
    Design an event-triggered control strategy (an objective-based triggering condition $\Xi$ and a deadline $\delta_{\max}$) to maximize the DIET function~$\bar\tau_\gamma(x_0)$ for all $x_0$  while maintaining $\Xi(\cdot)<0$ at all time.~\hfill$\bullet$
\end{problem}
Notice that our starting point in the problem is a given objective-based triggering condition. Theoretically, there may exist a triggering condition that can directly achieve our goals, but even if it exists, constructing such condition is extremely difficult. 
This paper provides a way to modify an existing triggering condition to improve the DIET metric through using a hierarchical framework. 

\section{A Hierarchical Architecture}
\begin{figure}
    \centering
    \includegraphics[width=0.8\columnwidth]{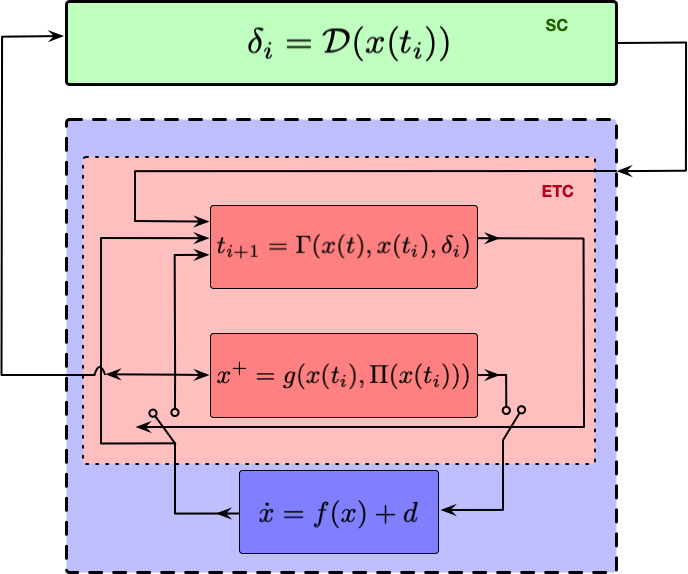}
    \caption{Schematic description of the layered ETC architecture}
    \label{fig:enter-label}
\end{figure}

We propose a two-layer hierarchical framework to simplify Problem~\ref{problem:diet} by decoupling the optimization of the DIET from satisfying the control objectives. 
The lower layer guarantees safety by leveraging the existing triggering condition to maintain safety of the control objectives. At the same time, we modify the ETC mechanism to include deadlines in order to reduce the greediness of the traditional ETC. The values of these deadlines are determined at the higher layer, and are adjusted each time a jump occurs. The higher layer aims to optimize the DIET by implementing a deadline selection policy for the event-triggered system.

\subsection{Safety layer - satisfying control objectives}
By definition, given a Zeno-free objective-based triggering condition $\Xi$, we can directly ensure the satisfaction of the control objectives by guaranteeing $\Xi$ remains negative across time. Standard ETC approaches extend inter-event times as much as possible for each immediate event, essentially implementing a greedy policy towards reducing the overall number of events. In general, there may exist a triggering condition $\Xi'$ that has a shorter immediate inter-event time $\tau_{\Xi'}(x)<\tau_\Xi(x)$, but whose DIET $\bar \tau_\gamma$ is larger because it leads the system towards better (with longer deadlines) triggering states $\{x(t_i)\}_{i\in\naturals}^{\infty}$. To ensure that the control objectives are maintained, the new condition $\Xi'$ must be dominated by the original condition $\Xi$. 
\begin{definition}\longthmtitle{Dominating Triggering Condition}\label{def:domination}
    A triggering condition $\Xi$ dominates a triggering condition $\Xi'$ if 
    $\Xi'(x,\bar x)<0 \implies \Xi(x,\bar x)<0$ for all $x,\bar x\in\real^n$.~\hfill$\bullet$
\end{definition}
The task of designing a dominated triggering condition $\Xi'$ that optimizes the DIET can generally be very challenging, so this paper takes a different approach to optimizing the DIET. Inspired by the idea of enforcing early triggering for periodic ETC, put forth by~\cite{RR:22}, we enrich the ETC scheme to include dynamic deadlines as:
\begin{align}\label{eq:trigger_deadline}
    t_{i+1} &= \min\{\min\setdefb{t\geq t_i}{\Xi(x(t),x(t_i))\geq 0}, t_i+ \delta_i\}\nonumber\\
    &:= \Gamma(x(t),x(t_i),\delta_i),
\end{align}    
where $\delta_i$ are the deadlines available to the upper layer as a control input to affect the triggering condition. The following result shows that we can recover any dominated triggering condition through the use of deadline policies $\map{\Dc}{\Xc}{\realnonneg}$ to determine these deadline inputs.
\begin{proposition}\longthmtitle{Equivalent Deadline Policy}
\label{prop:equiv}
    Consider the triggering conditions $\Xi$ and $\Xi'$ such that $\Xi$ dominates $\Xi'$. Then, there exists a deadline policy $\map{\Dc}{\Xc}{\realnonneg}$ such that the ETC scheme~\eqref{eq:trigger} using $\Xi'$ and the scheme~\eqref{eq:trigger_deadline} using $\Xi$ with the feedback $\delta_i=(\Dc(x(t_i))$ determine the same triggering times $\{t_i\}_{i=0}^\infty$.
\end{proposition}
\begin{proof}
    By Definition~\ref{def:domination}, the condition $\Xi(x(t),x(t_i))\geq 0$ cannot be met if $\Xi'(x(t),x(t_i)) < 0$. Therefore, the inter-event time function is bounded as $\tau_{\Xi'}(x)<\tau_\Xi(x)$ for all $x\in \Xc$.
    Let the deadline policy $\Dc(x) := \tau_{\Xi'}(x)$ be the inter-event times of $\Xi'$. It follows that \eqref{eq:trigger_deadline} evaluates to $t_{i+1} = \tau_{\Xi'}(x)$. Hence the two ETC schemes produce the same sequence $\{t_i\}_{i\in\naturals}^\infty$, concluding the proof.
\end{proof}
Proposition~\ref{prop:equiv} suggests we can use the deadline policy as a proxy for obtaining the optimal triggering strategy. However, we need to rule out the possibility of a deadline policy generating an invalid triggering condition that does not enforces control objectives.  

\begin{proposition} \longthmtitle{Triggering with Deadlines Enforces Control Objectives}\label{prop:obj}
    Consider the ETC scheme~\eqref{eq:trigger_deadline} using a triggering condition $\Xi$ and any deadline policy $\map{\Dc}{\Xc}{[\delta_{\min},\delta_{\max}]}$ for feedback $\delta_i=\Dc(x(t_i))$ with $\delta_{\min} >0$. Then, there exists an equivalent ETC scheme in the form of~\eqref{eq:trigger} using a $\Xi'$ such that $\Xi$ dominates $\Xi'$. As a consequence, the trigger scheme enforces $\Xi(x(t),x(t_i))<0$ at all time.
\end{proposition}
\begin{proof}
    The ETC scheme~\eqref{eq:trigger_deadline} with the deadline policy $\Dc$ determines the triggering time as:
    $$
    t_{i+1} = t_i + \min\{\tau_\Xi(x(t_i)),\Dc(x(t_i))\}.
    $$
    Therefore, the  ETC scheme~\eqref{eq:trigger} using he triggering condition $\Xi'(\cdot)=\min{\tau_\Xi(x(t_i)),\Dc(x(t_i))}$ is equivalent.  Then the trigger design does not exhibit Zeno behavior because it has a MIET $\min\{\tau^*,\delta_{\min} \}$ due to $\Dc(x)\geq \delta_{\min}$. Consequently, $\Xi(\cdot)<0$ at all time because the trigger can only occur before $\tau_\Xi(x(t_i))$.
\end{proof}

Note that, similarly to Assumption~\ref{assump:MIET}, we impose a restriction on the deadlines with a positive lower bound~$\delta_{\min}$ in order to restrict ourselves to non-Zeno trigger conditions with a MIET.  The implication of Proposition~\ref{prop:equiv} is that we decouple Problem~\ref{problem:diet} into designing a trigger to enforce control objectives and designing the deadline policy to optimize for the DIET. Regarding the former, the following remark discusses one of the main benefits of our framework.
\begin{remark}\longthmtitle{Greedy triggering conditions}
The effectiveness of our framework relies on a greedy nominal triggering condition. If the original triggering condition~\eqref{eq:trigger} enforces a large inter-event time, the optimal DIET value will improve due to an increase in available choices of effective deadlines~$\delta$. Therefore, our proposed framework incentivizes objective-based triggering conditions that extend each inter-event time as much as possible. To this end, the trigger design task is reduced to finding the least conservative condition that guarantees the control objective.
~\hfill$\bullet$
\end{remark}

\subsection{Optimization layer - improving DIET}
With the control objectives guaranteed by the lower layer, the event-triggered system can now be  viewed from the higher layer simply as a discrete-time control system with the deadline $\delta\in[\delta_{\min},\delta_{\max}]$ as an input:
\begin{equation}\label{eq:abstraction}
x_{i+1}  = G(x_i,\delta_i),
\end{equation}
where $\map{G}{\Xc\times [\delta_{\min},\delta_{\max}]}{\Xc}$ and $x_{i+1} = x(t_i)$. This system abstracts away the event-times $\{t_i\}_{i\in\naturals}^\infty$ in the lower layer. Nonetheless, the inter-event times can be exposed through a \textit{reward} associated to each discrete time-step:
\begin{equation}\label{eq:reward}
r(x,\delta):=\min\{\tau_\Xi(x),\delta\}.
\end{equation}
From the perspective of the optimization layer, we reformulate the optimization of DIET as an optimal control problem.
\begin{problem}\longthmtitle{Optimization layer optimal control problem}
    \label{problem:opt}
    Consider the discrete time system \eqref{eq:abstraction}-\eqref{eq:reward}. Find the optimal deadline policy $\map{\Dc^*}{\Xc}{[\delta_{\min},\delta_{\max}]}$ that maximizes the DIET $\bar \tau_\gamma(x_0) = \liminf_{s\rightarrow\infty}\sum_{i=0}^s\gamma^ir(x_i,\Dc^*(x_i))$ for all initial conditions $x_0\in\Xc$ .~\hfill$\bullet$
\end{problem}

The main obstruction to Problem~\ref{problem:opt} is the map $G$ being nonlinear and difficult to model. Obtaining the map $G$ capturing accurately the corresponding evolution of $x(t_i)$ at event-times is extremely challenging, particularly for general nonlinear systems. Symbolic abstractions have been proposed to this end for both linear~\cite{ASK-MM:16,GG-MM:21} and nonlinear~\cite{GD-MM:22} systems. However, these abstractions can be computationally demanding to construct and/or conservative, particularly in the case of nonlinear (disturbed) systems. Therefore, instead of constructing such a model and design a (model-based) optimal deadline control policy, we propose to employ a model-free alternative approach.

\subsection{Reinforcement learning for solving Problem~\ref{problem:opt}}
Reinforcement learning (approximately) solves optimal control problems by exploring different available actions and observing the obtained rewards to update the control policy. In our context, the reward the agent receives after executing an action is already clearly defined. In this paper, we restrict the presentation to the simple $Q$-learning method, which we review in what follows.

We start with a discretization of our state space $\Xc$ and action space $[\delta_{\min},\delta_{\max}]$ into the finite sets $\Sc$ and $\Ac$, respectively, so that:
$$
\mathbf{x}_{i+1} = \mathbf{G}(\mathbf{x}_i,\textbf{\textdelta}_i) + \mathbf{e}
$$
with new state $\mathbf{x}\in \Sc$, deadline $\mydelta\in \Ac$, map $\map{\mathbf{G}}{\Sc\times \Ac}{\Sc}$, and the (nondeterministic) error $\mathbf{e}$ due to this discretization. While technically this error is not stochastic, we treat it as if it would be and use this model to approximate the optimal deadline policy. Under this stochastic perspective, we denote by $P_{\mathbf{x}\mathbf{x}'}(\mydelta)$ the probability of $\mathbf{G}(\mathbf{x},\textbf{\textdelta}) + \mathbf{e}=\mathbf{x}'$, and analogously we denote by $R(\mathbf{x},\mydelta)$ the random variable associated to the observed reward $r(x,\mydelta)$, where $\mathbf{x}$ is the discretized state associated to the actual state $x$.

The $Q$-learning algorithm associates a \textit{value} function to each state $\mathbf{x}$ under a policy $\map{\tilde \Dc}{\Sc}{\Ac}$:
$$V^{\tilde\Dc}:(\mathbf{x}):=\Exp[R(\mathbf{x},\tilde\Dc(\mathbf{x}))] + \gamma \sum_{\mathbf{x}'\in\Ac}P_{\mathbf{x}\mathbf{x}'}(\tilde\Dc(\mathbf{x})) V^{\tilde\Dc}(\mathbf{x}'),$$
i.e., the sum of an expected instantaneous reward plus the expected discounted value of the next state. 
The optimal value function $V^*$ represents the optimal (expected) discounted sum of rewards from each state.  By the Bellman principle of optimality, $V^*$ satisfies: $V^*(\mathbf{x})=\max_{\mydelta\in\Ac}\lbrace{\Exp[R(\mathbf{x},\mydelta)] + \gamma \sum_{\mathbf{x}'\in\Ac}P_{\mathbf{x}\mathbf{x}'}(\delta) V^*(\mathbf{x}')\rbrace}$. 
The goal in Q-learning is to approximate a function $\map{Q^*}{\Sc\times \Ac}{\realpos}$ which when optimized over the actions yields the optimal policy, i.e. $V^{\tilde\Dc^*}=V^*$, with 
\begin{equation}\label{eq:optimal_deadline}
    \tilde\Dc^*(\mathbf{x}):=\argmax_{\mydelta\in\Ac}Q^*(\mathbf{x},\mydelta).
\end{equation}
Furthermore, the optimal $Q^*$ satisfies $Q^*(\mathbf{x},\mydelta)=\Exp[R(\mathbf{x},\mydelta)]+\gamma \sum_{\mathbf{x}'\in\Ac}P_{\mathbf{x}\mathbf{x}'}(\delta) V^*(\mathbf{x}')$.

In $Q$-learning, the optimal $Q^*$ is approximated iteratively through experimentation. Starting from a $Q_0$ arbitrarily initialized, after each observed transition the approximation of $Q^*$ is adjusted employing the observed reward following:
\begin{equation}\label{eq:Q-update}
Q_{i+1}(\mathbf{x},\mydelta) = (1-\alpha_i)Q(\mathbf{x},\mydelta) + \alpha_i\big(R(\mathbf{x},\mydelta)+\gamma V_i(\mathbf{x}')\big)
\end{equation}
where $\mathbf{x}'=G(\mathbf{x},\mydelta)+\mathbf{e}$ (the observed next state) $V_i(\mathbf{x}')=\max_{d}Q_i(\mathbf{x}',d)$, and $\alpha_i\in(0,1)$ is the (possibly time varying) learning rate. Under certain conditions of boundedness of rewards and quadratic sum convergence for the rates $\alpha_j$ iterating $Q$ according to \eqref{eq:Q-update} leads to the convergence~\cite{PD-CJCHW:92} to the optimal $Q^*$, $\lim_{j\to\infty}Q_j\rightarrow Q^*$.
Then, the ``optimal" deadline policy~$\tilde \Dc^*$ can be recovered via the relationship~\eqref{eq:optimal_deadline}. 

The $Q$-learning algorithm yields an approximation to the optimal deadline policy for our original problem, as we employ a discretized model approximation, and in practice, the algorithm must be terminated after a finite number steps.

\begin{remark}\longthmtitle{Safe online learning}
One of the strong features of our framework is that the learning can be safely implemented online because Proposition~\ref{prop:obj} guarantees $\Xi(\cdot)>0$ regardless of the deadline policy $\Dc_j$ being used. This opens many doors for using transfer learning approaches that allow the users to train a policy offline with a reduced model and then adapt online the policy to the real dynamics.
~\hfill$\bullet$
\end{remark}

We propose Algorithm~\ref{alg:updateQ} for $Q$-learning specific to our deadline policy, in order to speed up the learning process. The idea is based on the fact that when an action $\mydelta$ is taken at $\mathbf{x}$, we can also observe along the trajectory the rewards $R(\mathbf{x},\mydelta')$ and the states $\mathbf{x}^{\mydelta'}_{i+1}=\mathbf{x}(t_i+\mydelta')$ for actions $\mydelta'<\mydelta$. In addition, if the event is dictated by $\Xi$ rather than the deadline condition, we can also record the same rewards for $\mydelta'>\mydelta$. Thus, we may update the function $Q(\mathbf{x}_i,\mydelta')$  for an array of actions $\mydelta'$ for the state $\mathbf{x}_i$ with each triggering event, even when only one action~$\delta$ is actually taken.
\begin{algorithm}
\caption{Updating $Q$ for multiple actions per trigger in deadline learning}
\label{alg:updateQ}
 \hspace*{\algorithmicindent} \textbf{Input:} $Q_i$, $t_i$, $t_{i+1}$, \mydelta, $[t_i,t_{i+1}]\mapsto x(t)$\\
 \hspace*{\algorithmicindent} \textbf{Output:} $Q_{i+1}$ 
\begin{algorithmic}[1]
    \For{$\mydelta'\in\Ac$ and $\mydelta' \leq t_{i+1}-t_i$}
        \State $R\gets \mydelta'$ 
        \State $V \gets \max_{a\in\Ac} Q(x(t_i+\mydelta),a)$
        \State $Q_{i+1}(\mathbf{x},\mydelta') \gets(1-\alpha)Q_i(\mathbf{x},\mydelta')+\alpha(R+\gamma V)$
    \EndFor
    \If{$t_{i+1}-t_i< \mydelta$}
    \For{$\mydelta'\in\Ac$ and $\mydelta'> t_{i+1}-t_i$}
        \State $R\gets t_{i+1}-t_i$
        \State $V \gets \max_{a\in\Ac} Q(x(t_{i+1}),a)$
        \State $Q_{i+1}(\mathbf{x},\mydelta') \gets(1-\alpha)Q_i(\mathbf{x},\mydelta')+\alpha(R+\gamma V)$
    \EndFor
    \EndIf
\end{algorithmic}
\end{algorithm}

\section{Application to Satellite Orbiting Control}
We demonstrate our hierarchical framework with a satellite orbit safety problem.

\subsection{An orbit safety problem}
The satellite's dynamics are given by the Newton's gravitational model:
$$
\frac{d}{dt}\begin{bmatrix}\vec{r}\\ \vec{v}\end{bmatrix} = \begin{bmatrix}\vec{v}\\- \frac{\mu}{r^3} \vec{r}\end{bmatrix}+\begin{bmatrix}
    0\\ d
\end{bmatrix},
$$
where the states $\vec{r}(t)\in\real^3$ and $\vec{v}(t)\in\real^3$ are the satellite's position and velocity, respectively and $\mu>0$ is the \textit{gravitational parameter} of the central body of the orbit. We use the short hand notation $r=\|\vec{r}\|$ for the satellite's orbital radius. Due to the disturbance $d(t)\in\real^3$ in the acceleration, the satellite may drift away from a safe orbit. In response, it can actuate its thruster to change its velocity. The effect of the thruster is typically modelled as an impulse due to the difference in the timescale between the thruster actuation and satellite free-flow orbit (seconds vs. days):
$$
\begin{bmatrix}\vec{r}\\ \vec{v}\end{bmatrix}^+ = \begin{bmatrix}\vec{r}\\ \vec{v}\end{bmatrix}+ \begin{bmatrix}0\\ \Delta \vec{v}\end{bmatrix}.
$$
where $\Delta \vec{v}\in\real^3$ is the control input. 

We consider a satellite orbiting the asteroid 25143 Itokawa. The satellite is subjected to disturbances from the higher-order gravity field (state-dependent model uncertainty) and the effect of the central body rotating on its axis (time-varying disturbance). Our control objective is to maintain the satellite within safe radius range,  $1.6R<r<2.4R$ where $R$ is the asteroid's mean body radius.

\subsection{Safety layer setup}
We use a barrier function $h(r) = (0.4R)^2 - (r-2R)^2$ to define the set of safe states as:
\begin{align*}
    \Cc &= \setdefb{ (\vec{r},\vec{v}) \in \real^3\times \real^3}{h(r) \geq 0}. 
\end{align*}
The barrier condition~\cite{ADA-SC-ME-GN-KS-PT:19} for safety is then given by:
$$
\underbrace{-2(r-0.4R)(\vec{r}\cdot\vec{v})/r}_{= \Lie_fh(\vec{r},\vec{v})} -\left\|\left. \frac{\partial h}{\partial x} \right|_{x}d\right\|\geq -\alpha(h(r))
$$
where we select $\alpha(z) = z/600$. Furthermore, we assume the term with the disturbance is bounded, and we propose the nominal triggering condition as:
$$
\Xi(\vec{r},\vec{v}) = -1200(r-0.4R)(\vec{r}\cdot\vec{v})/r + h(r) - 0.0005 
$$

In regard to the impulse feedback control policy $\Delta v = \Pi(\vec{r},\vec{v})$, we directly calculate the desired $\vec{v}^+$ using orbital mechanics as follows. First, we calculate a ``safety factor" $S = \frac{r-2R}{0.4R}$. We use it to compute a target radius of the \textit{peri/apoapsis} of a new orbit that the satellite will be injected into:
$$
r_\tgt = 2R-(0.2R)S=3R-0.5r.
$$
Similarly, we select a \textit{true anomaly} for the satellite in the new orbit according to:
$$
\nu = 
\begin{cases}
     -\pi + \pi S,& S>0 \\
     -(\pi/2)S, & S\leq 0
\end{cases}
$$
The reason to pick the target radius and the \textit{true anomaly} with a bias based on $S$ is so that the jump will result in a positive $\Lie_fh(\vec{r},\vec{v})$. In turn, the barrier condition is satisfied strictly, ruling out Zeno behavior. Based on $r_\tgt$ and $\nu$, the \textit{eccentricity} and the \textit{semi-latus rectum} can be computed:
\begin{align*}
    e &= (r-r_\tgt)(\sgn(S) r_\tgt-r\cos\nu)\\
    p &= r_\tgt(1+\sgn(S) e).
\end{align*}
Then the desired velocity is given by:
$$
\vec{v}^+ = \sqrt{\mu/p} (-\sin\nu\vec{P}+(e+\cos\nu)\vec{Q})
$$
where $\vec{P},\vec{Q}$ are the basis vector of the \textit{perifocal coordinate system} \cite[Sec. 2.2.4]{RRB-DDM-JEW:71}. Note that $\Delta v = \vec{v}^+-\vec{v}$.

\subsection{Optimization layer setup}
Now that we have the safety layer enforcing the control safety objective, we can focus on resource conservation via the optimization layer. Here we are concerned with reducing the number of thruster usages because they may interfere with the satellite's mission or generate unnecessary heat. We note that we do not yet consider in this paper the total propellant uses from applying $\Delta v$, but reducing the frequency of actuation is nevertheless a good proxy for it. Future work will include propellant usage in the optimization.

\begin{remark}\longthmtitle{Reward function and actions}
    The reward function considered in this paper only depends on the inter-event times. In the future, we want to incorporate the state $x$ and the impulses $v$. This would allow us to pose richer optimal control problems, as well as including the impulsive control input $v$ in the optimization.~\hfill$\bullet$ 
\end{remark}
For the learning of the deadline policy, we discretize the states into 400 buckets based solely on the radius $r$ and the actions into 10000 different deadlines, ranging from 50 seconds to 100 hours, spaced exponentially. We update the $Q$ function over 180 generations of 100 episodes. Each episode simulates a trajectory from a random initial condition until 20 events take place. Fig.~\ref{fig:DIET-log} shows the improvement of DIET over generations. We note that because the plot displays the DIET in logarithmic scale, the improvement in the observed maximum DIET over generation is downplayed and the decrease in minimum DIET observed is exaggerated. 

\subsection{Results}

To test the efficiency of our learned deadline policy, we ran simulations of the closed-loop system from 100 different initial conditions. For each initial condition, we recorded the DIET over 50 events, for both the traditional ETC (greedy) policy and the learned policy. We report the average (across the 100 trajectories) of the DIET to be 2653 and 4422 hours per trigger, respectively, suggesting an improvement of roughly 1.7 times overall. 

We plot the learned deadline policy for each of the 400 state buckets in Fig. \ref{fig:deadline_policy}, and we notice that the deadlines differ drastically from the greedy approach (using the maximum deadline) for radii closer to the boundary of the safe set. We ran 100 simulations for a fixed initial radius of 2.3R where the deadline policy drastically change. We report the average of the DIET as 110 and 4644 hours for the greedy policy and the learned policy, respectively. Here, the improvement is staggeringly by the factor of 42. We conclude that the proposed hierarchy can provide great improvement in specific regions.
The plots of 10 system trajectories, cf. Fig.~\ref{fig:learned_behavior}, shows that system safety is satisfied. The learned policy has learned that the inter-event time for states away from the boundary are naturally longer, so it takes advantage by triggering early than safety requires, in order to dwell at the states away from the boundary.

\begin{figure}
    \centering
    \includegraphics[width=\columnwidth]{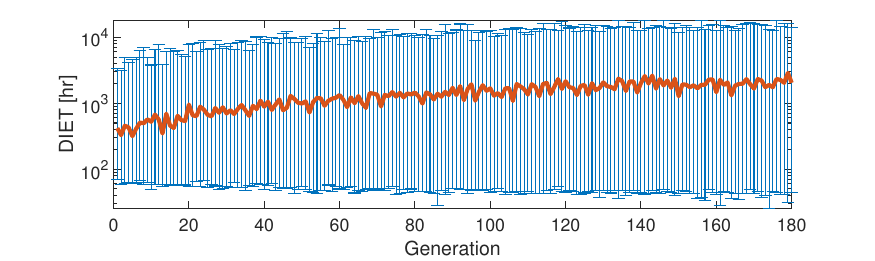}
    \caption{Improvement of DIET over generations of learning. For each generation, we plot the average (red) of DIET of the 100 trajectories within the generation, and the error bar (blue) encompasses the range of DIET observed. The DIET axis is displayed in logarithmic scale to also reveal the worst-case.}
    \label{fig:DIET-log}
\end{figure}

\begin{figure}
    \centering
    \includegraphics[width=\columnwidth]{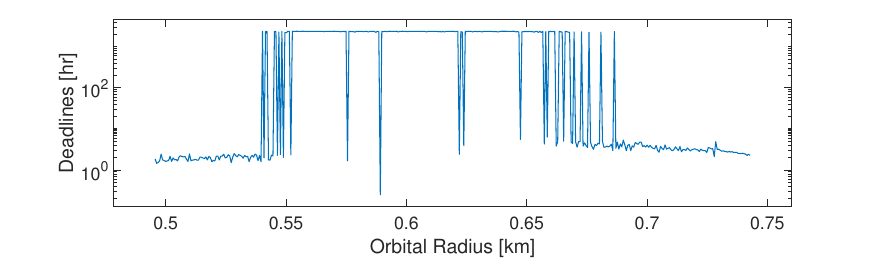}
    \caption{Learned deadline policy. The plot shows the deadline values for each of the 400 state buckets. For radii close to the boundary of the safe set, the learned deadline is much lower than 100 hours, which is used by the traditional ETC (greedy) policy.}
    \label{fig:deadline_policy}
\end{figure}
\begin{figure}
    \centering
    \includegraphics[width=\columnwidth]{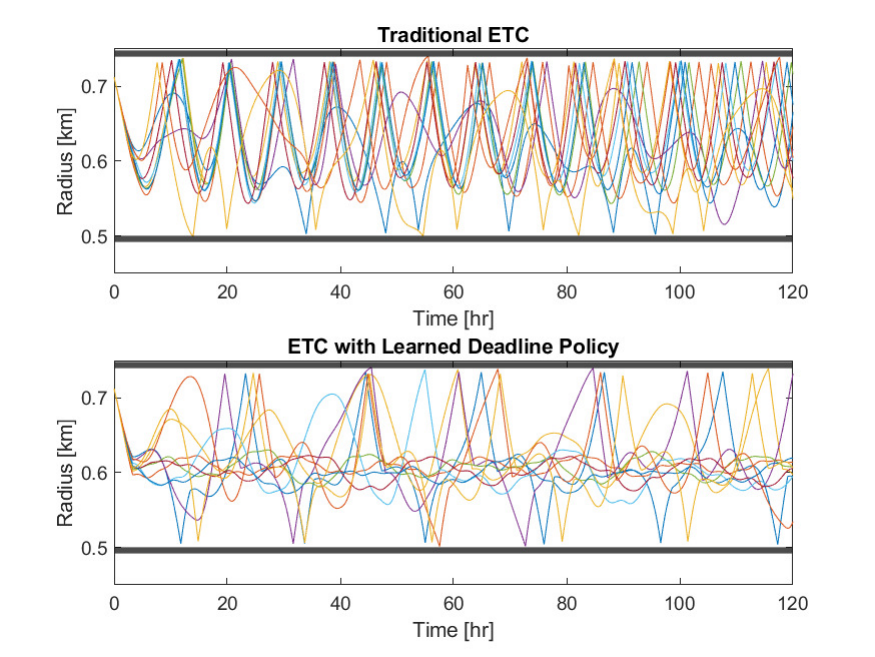}
    \caption{Ten trajectory comparisons between the greedy deadline policy (top) and learned deadline policy (bottom). In both cases, the underlying ETC strictly enforces safety. However, the learned policy intelligently triggers in the way that trajectories dwell in the region where inter-event times are longer, leading to an overall increase in DIET. }
    \label{fig:learned_behavior}
\end{figure}

\section{Conclusion}

We have presented a hierarchical framework for optimizing ETC trigger designs for less resources' spending. The framework allows the users to employ any available ETC design, retaining its corresponding guarantees. Furthermore, our proposal decouples the control objectives, guaranteed by the ETC design, from the optimization of long-term inter-event times' metrics via deadline policies. To address the latter, an optimal control problem over abstracted dynamics, not explicitly constructed, is proposed and proposed to be solved via model-free methods, as e.g., RL. In addition, we exploited the structure of deadline generation in a RL implementation with improved learning speed.
Future work includes exploration of different learning algorithms, including online transfer learning, richer objective functions in the higher layer optimal control problem, and the incorporation
of stochastic noise in the framework.

\end{document}